\documentclass{article}
\usepackage{latexsym}
\usepackage{verbatim}
\newtheorem{theorem}{Theorem}

\newtheorem{corollary}{Corollary}

\newcommand{\blackslug}{\mbox{\hskip 1pt \vrule width 4pt height 8pt 
depth 1.5pt \hskip 1pt}}
\newcommand{\qed}{\quad\blackslug\lower 8.5pt\null\par\noindent}
\newenvironment{proof}{\par\noindent{\bf Proof:}}{\qed \par}

\newcommand{\bH}{\mbox{${\bf H}$}}
\newcommand{\bG}{\mbox{${\bf G}$}}

\newcommand{\eqdef}{\stackrel{\rm def}{=}}

\title{Measurements and majorization
}

\author{Daniel Lehmann\\School of Engineering
and \\ Center for Language, Logic and Cognition
\\Hebrew University, \\Jerusalem 91904, Israel 
}
\date{February 2012}

\begin{document}
\maketitle
\begin{abstract}
Majorization is an outstanding tool to compare the purity of mixed states
or the amount of information they contain and also the
degrees of entanglement presented by such states in tensor products. 
States are compared by their spectra and majorization defines a partial order 
on those. 
This paper studies the effect of measurements on the majorization relation 
among states. It, then, proceeds to study the effect of local measurements
on the agents sharing an entangled global state.
If the result of the measurement is recorded, Nielsen and 
Vidal~\cite{Nielsen-Vidal:majorization} showed that the expected spectrum
after any P.O.V.M. measurement majorizes the initial spectrum, 
i.e., a P.O.V.M. measurement
cannot, in expectation, reduce the information of the observer. 
A new proof of this result is presented and, as a consequence, the {\em only if} part
of Nielsen's~\cite{Nielsen:LOCC} characterization of LOCC transformations is 
generalized to $n$-party entanglement.
If the result of a bi-stochastic measurement is not recorded, 
the initial state majorizes the final
state, i.e., no information may be gained by such a measurement.
This strengthens a result of A. Peres~\cite{Peres:QuantumTheory}.
In the $n$-party setting, no local trace preserving measurement by Alice 
can change the local state of another agent.
\end{abstract}

\section{Introduction} \label{sec:intro}
In this paper {\em measurement} means P.O.V.M. measurement, i.e., a family
\mbox{$\{ f_{k} \}_{k = 1}^{m}$} of operators \mbox{$f_{k} : \bH \rightarrow \bH$} such
that 
\begin{equation} \label{eq:povm}
\sum_{k = 1}^{m} f_{k}^{\ast} \circ f_{k} \ = \ {id}_{\bH}.
\end{equation}
A bi-stochastic measurement is a measurement that also satisfies 
\begin{equation} \label{eq:bistoch}
\sum_{k = 1}^{m} f_{k} \circ f_{k}^{\ast} \ = \  {id}_{\bH}.
\end{equation}
The term {\em state} means mixed state, i.e., a self-adjoint, weakly positive operator
\mbox{$\rho : \bH \rightarrow \bH$} of trace $1$.
All Hilbert spaces are finite-dimensional.
The majorization relation is written $\succeq$ both for vectors of 
non-negative real numbers and for self-adjoint operators.
We shall only use the majorization relation to compare two vectors that sum
up to the same value (often $1$ but not always)  or two self-adjoint operators
of equal traces (often $1$ but not always). 
The entropy of a state $\sigma$ is denoted $E(\sigma)$.
The spectrum of a state $\sigma$ is denoted $Sp(\sigma)$.
The projection on $x$ is denoted \mbox{$\mid \! \! x \rangle \langle x \! \! \mid$}.
The conjugate of a complex number $c$ is denoted $\overline{c}$.

The basic question we study is the following: how does a measurement change
the {\em purity} of a quantic state? 
The understanding is that when the observer knows that a quantic system
is in a specific pure state, he knows everything that can possibly be known about 
the system. 
If, for all he knows, the system is in a state that is not pure, he is 
somehow uncertain about the state of the system. 
The purpose of a measurement is to gather information about the system 
and therefore one expects that the state resulting from a measurement will be
{\em purer} than the initial state. 
This paradigm fits the classical intuition, but, in QM, caveats must be addressed.
Consider, first, a measurement the result of which is recorded.
Classically, the effect of such a measurement on the observer's 
information depends on its result and the observer may find himself in a more 
uncertain situation after the measurement, but, 
in expectation, the observer's uncertainty cannot increase.
One expects the same to hold in QM.
Consider, now, a measurement the result of which is not recorded.
Classically, such a measurement does not change the observer's knowledge 
and therefore does not change the uncertainty attached to the system. 
In QM such a measurement, which can be realized, for example, 
by sending a particle onto one of two different paths that are then merged,
if no witness exists of which path has been taken, may transform a pure state into
a state that is not pure and therefore increase the observer's uncertainty.

The discussion above has not fixed the way one should measure {\em purity}
or {\em uncertainty}.
A lot of effort has been devoted to measuring qualitatively and quantitatively
the {\em purity} of quantic states. 
Recently, Nielsen showed the relevance of the 
majorization partial order in Quantum Information:
see~\cite{Marshall-Olkin:Inequalities} on majorization and~\cite{Nielsen:intro} for
its relevance to QI. Nielsen credits Uhlmann~\cite{Uhlmann:Entropy} for noticing
the link between majorization and QM.
This paper considers the effect of quantic measurements on the majorization
partial order among states. It, then, applies the results obtained to the study of
the effect of local measurements to local states of an entangled system.

\section{The effect of measurements} \label{sec:povm-measurements}
\subsection{Spectra} \label{sec:spectra}
If one decides to measure some observable defined by a family
\mbox{$\{ f_{k} \}_{k = 1}^{m}$} 
one will obtain some result, i.e., some $k$, for one's measurement.
The initial state of the system defines a probability distribution 
on the possible states that can be the result of the measurement, 
not a single state.
If the initial state is $\sigma$ then the probability of obtaining result $k$ 
is given by:
\begin{equation} \label{eq:prob-gen}
p_{k} \ = \ Tr ( f_{k}  \circ \sigma \circ {f}_{k}^{\ast}).
\end{equation}
If \mbox{$p_{k} =$} $0$, the result $k$ is never obtained.
If \mbox{$p_{k} >$} $0$, the state that results from the measurement is given by:
\begin{equation} \label{eq:new-state}
\sigma'_{k} \ = \ {{1} \over {p_{k}}} \ f_{k} \circ \sigma \circ f_{k}^{\ast} .
\end{equation}

Our purpose is to compare the spectrum of the state before the measurement 
and the spectrum after the measurement. 
There are two spectra that come to mind as candidates for the spectrum
after the measurement.
The first one is the spectrum of the state
\mbox{$\sum_{k = 1}^{m} \, p_{k} \, \sigma'_{k} =$}
\mbox{$\sum_{k = 1}^{m} \, f_{k} \circ \sigma \circ f_{k}^{\ast}$}.
In the literature, when this is the result considered, 
one calls the measurement a {\em trace-preserving} measurement.
But there is another spectrum that can be considered 
to be the result of the measurement:
the expected spectrum defined as a convex combination 
of the spectra of the states \mbox{$\sigma'_{k}$} 
where the spectrum of $\sigma'_{k}$ is weighted by the probability $p_{k}$
of obtaining the state $\sigma'_{k}$. 
For defining this combination, we consider a spectrum 
to be a {\em decreasing} vector of non-negative real numbers 
and add {\em componentwise}.
In other term, the $k$'th largest value of the expected spectrum is 
the expected value of the $k$th largest values obtained in the different possible
outcomes and the spectrum considered is:
\mbox{$\sum_{k = 1}^{m} \, p_{k} \, Sp ( \sigma'_{k} ) =$}
\mbox{$\sum_{k = 1}^{m} \, Sp ( f_{k} \circ \sigma \circ f_{k}^{\ast} )$}.
When such a result is considered, the  literature calls the measurement an 
{\em efficient} measurement.

\sloppypar
In summary we want to compare $Sp( \sigma )$,
\mbox{$\sum_{k = 1}^{m} \, Sp( f_{k} \circ \sigma \circ f_{k}^{\ast} )$} and
\mbox{$Sp ( \sum_{k = 1}^{m} \, f_{k} \circ \sigma \circ f_{k}^{\ast} )$}.
We shall show that, for any measurement, the final expected spectrum majorizes 
the initial spectrum:
\mbox{$\sum_{k = 1}^{m} \, Sp( f_{k} \circ \sigma \circ f_{k}^{\ast} ) \succeq$}
\mbox{$Sp( \sigma )$} and that, for any {\em bi-stochastic} measurement:
\mbox{$Sp( \sigma ) \succeq$} 
\mbox{$Sp ( \sum_{k = 1}^{m} \, f_{k} \circ \sigma \circ f_{k}^{\ast} )$}.

If we consider entropy, a possible measure of information, one notes that there
are two natural ways of measuring the entropy resulting 
from an efficient measurement. 
One may consider the entropy of the expected spectrum
defined above, but one may also consider the expected entropy.
The former quantity is defined by 
\mbox{$S_{1} =$}
\mbox{$ S ( \sum_{k = 1}^{m} \, p_{k} \, Sp ( \sigma'_{k} ) ) $} and the latter by
\mbox{$S_{2} =$}
\mbox{$ \sum_{k = 1}^{m} \, p_{k} \, S ( Sp ( \sigma'_{k} ) ) $}.
The result below implies that \mbox{$S_{1} \leq S ( \sigma )$}.
The concavity of entropy implies \mbox{$S_{2} \leq S_{1}$}.
The result above therefore implies a definite strengthening of the fact that, in 
expectation, the entropy cannot be increased by an efficient measurement.

\section{Efficient measurements} \label{sec:efficient}
\subsection{Past work} \label{sec:pastexpect}
Theorem~\ref{the:expect} below is 
Theorem~12 of~\cite{Nielsen-Vidal:majorization}.
The authors present the result as a corollary of the characterization of LOCC 
transformations in $2$-party systems in a pure state obtained 
in~\cite{Nielsen:LOCC}. 
The $1$-party result is proved by reduction to the $2$-party result by purification. 
The proof presented below is a direct proof.

We shall rely on two results.
The first is a corollary of a theorem of Y. Fan (Theorem~1 of~\cite{Fan:49}).
It may be found in~\cite{Marshall-Olkin:Inequalities} p. 241.
\begin{theorem} \label{the:sum}
For any self-adjoint matrices $A$ and $B$, 
\mbox{$A , B : \bH \rightarrow \bH$}, one has
\mbox{$Sp(A) + Sp(B) \succeq$} \mbox{$Sp ( A + B )$}.
\end{theorem}

The second one is most probably well-known, 
but no reference for it has been found.
\begin{theorem} \label{the:sp=}
Let $A$ be a finite dimensional Hilbert space and \mbox{$f : A \rightarrow A$} a linear operator.
Then, \mbox{$Sp ({f}^{\ast} \circ f ) =$}
\mbox{$Sp ( f \circ {f}^{\ast} )$}.
\end{theorem}
\begin{proof}
Note, first, that both ${f}^{\ast} \circ f$ and $f \circ {f}^{\ast}$ are self-adjoint
and therefore have $dim(A)$ real eigenvalues.
We shall show that every eigenvalue $\lambda$ of ${f}^{\ast} \circ f$, 
different from zero,
is an eigenvalue of $f \circ {f}^{\ast}$ with the same multiplicity.
To this effect we note that if \mbox{$x \in A$} is an eigenvector of ${f}^{\ast} \circ f$ 
for some eigenvalue \mbox{$\lambda \neq 0$}, then \mbox{$f ( x )$} is an eigenvector of 
$f \circ {f}^{\ast}$ for eigenvalue $\lambda$.
Suppose indeed $x$ and $\lambda$ are as assumed, then
\mbox{${f}^{\ast} ( f (x )) =$} \mbox{$\lambda \, x \neq \vec{0}$} and therefore
\mbox{$f ( x ) \neq \vec{0}$}. 
But \mbox{$(f \circ {f}^{\ast}) (f (x ) ) =$}
\mbox{$f ( ( {f}^{\ast} \circ f ) ( x) =$}
\mbox{$f ( \lambda \, x ) =$}
\mbox{$\lambda \, f ( x )$}.
We are left to show that the multiplicity of $\lambda$ for $f \circ {f}^{\ast}$ is at least its
multiplicity for ${f}^{\ast} \circ f$.
For this, we note that if \mbox{$y \in A$} is orthogonal to $x$, 
then $f ( y )$ is orthogonal to $f ( x )$.
Indeed, \mbox{$\langle f ( y ) \mid f ( x ) \rangle =$}
\mbox{$\langle y \mid ( {f}^{\ast} \circ f ) ( x ) \rangle =$}
\mbox{$ \langle y \mid \lambda \, x \rangle =$} 
\mbox{$\lambda \, \langle y \mid x \rangle = 0$}.
\end{proof}

\subsection{Result} \label{sec:expect-result}
\begin{theorem} \label{the:expect}
Let $\sigma$ be a state and \mbox{$\{ f_{k} \}_{k = 1}^{m} $} a measurement.
Then,
\begin{equation} \label{eq:expect}
\sum_{k = 1}^{m} \, Sp ( f_{k} \circ \sigma \circ f_{k}^{\ast} ) \succeq Sp ( \sigma ).
\end{equation}
\end{theorem}
\begin{proof}
The operator $\sigma$ is self-adjoint and weakly positive, it has therefore
a square root, i.e., a self-adjoint, weakly positive operator 
\mbox{$\alpha : \bH \rightarrow \bH$}
such that \mbox{$\sigma =$} \mbox{$\alpha \circ {\alpha}^{\ast}$}. 
Let \mbox{$\beta_{k} =$} \mbox{$f_{k} \circ \alpha$}. 
We have
\mbox{$f_{k} \circ \sigma \circ f_{k}^{\ast} =$} 
\mbox{$\beta_{k} \circ {\beta}_{k}^{\ast}$}. 
By Theorem~\ref{the:sp=} we have: 
\mbox{$Sp ( \beta_{k} \circ {\beta}_{k}^{\ast} ) =$} 
\mbox{$Sp ( {\beta}_{k}^{\ast} \circ \beta_{k} )$}.
We conclude that
\[
Sp ( f_{k} \circ \sigma \circ f_{k}^{\ast} ) \ = \ 
Sp ( {\alpha}^{\ast} \circ f_{k}^{\ast} \circ f_{k} \circ \alpha ).
\]
By Theorem~\ref{the:sum} and Equation~(\ref{eq:povm}) one has:
\begin{equation} \label{eq:Fan-Alice}
\sum_{k = 1}^{m} Sp ( \alpha^{\ast} \circ f_{k}^{\ast} \circ f_{k} \circ \alpha ) 
\ \succeq \ 
Sp ( \sum_{k = 1}^{m}  \alpha^{\ast} \circ f_{k}^{\ast} \circ f_{k} \circ \alpha) \ = \ 
Sp ( \alpha^{\ast} \circ \alpha ) \ = \ Sp ( \sigma ).
\end{equation}
\end{proof}

\section{Trace preserving measurements} \label{sec:no-record}
\subsection{Past work} \label{sec:Peres}
On p. 262, A. Peres~\cite{Peres:QuantumTheory} shows that a trace preserving 
von Neumann measurement cannot decrease entropy, 
or any concave function of the state for that matter. 
The meaning of such a result is that a measurement whose
result is not recorded can never increase the observer's information, but it can,
in fact, decrease this information.
We strengthen Peres' result on two counts.
First, by showing that the initial state majorizes the
final state. This indeed is a strengthening of Peres' result since Schur's
characterization of real functions that preserve majorization 
(Theorem 3.A.4 in~\cite{Marshall-Olkin:Inequalities}) implies that entropy and any
concave function considered by Peres, {\em anti-preserve} majorization:
\mbox{$\sigma \succeq \rho$} implies \mbox{$E(\sigma) \leq E(\rho)$}:
see, for example Proposition 4.2.1 in~\cite{Nielsen:intro}.
Then, the result is proved for any bi-stochastic measurement, 
not only von Neumann measurements.
 
\subsection{Result} \label{sec:vNresult}
The following is due to A. Uhlmann~\cite{Uhlmann:theorem}.
The proof given here for completeness' sake is streamlined from the proof
of Theorem 5.1.3 in~\cite{Nielsen:intro}.
\begin{theorem} \label{the:vNresult}
Let \mbox{$\sigma , \tau: \bH \rightarrow \bH$} be states. 
Then 
\begin{equation} \label{eq:vN}
\sigma \succeq \tau {\rm \ iff \ }
\tau \ = \ \sum_{k = 1}^{m} \: f_{k} \circ \sigma \circ f_{k}^{\ast}
\end{equation}
for some bi-stochastic measurement 
\mbox{$\{ f_{k} \}_{k = 1}^{m}$}. 

\end{theorem}
\begin{proof}
We first deal with the {\em if} direction.
Assume \mbox{$\{f_{k}\}_{k =1}^{m}$} is a bi-stochastic measurement.
Let us make, at first, the facilitating assumption that there is a basis 
\mbox{$\{ x_{i} \}_{i = 1}^{n}$} 
of eigenvectors of both $\sigma$ and
\mbox{$\tau \eqdef$}
\mbox{$\sum_{k = 1}^{m} \, f_{k} \circ \sigma \circ f_{k}^{\ast}$}.
Let 
\mbox{$\sigma ( x_{i} ) =$} \mbox{$\lambda_{i} \, x_{i}$} 
and \mbox{$\tau ( x_{i} ) =$} \mbox{$\mu_{i} \, x_{i}$} for any $i$
and let $\lambda$ (resp. $\mu$) be the real column vector \mbox{$[ \lambda_{i} ]$}
(resp. \mbox{$[ \mu_{i} ] )$}.

Define \mbox{$c_{i , j}^{k} =$} \mbox{$\langle x_{i} \mid f_{k} \mid x_{j} \rangle$}.
We have \mbox{$f_{k} ( x_{j} ) =$} \mbox{$\sum_{i = 1}^{n} \, c_{i , j}^{k} \, x_{i}$}.
We have, for any $i$, $j$:
\mbox{$\langle x_{j} \mid f_{k}^{\ast} \mid x_{i} \rangle \ = \ \overline{c_{i , j}^{k}}$}
and \mbox{$f_{k}^{\ast} ( x_{i} ) =$} 
\mbox{$\sum_{j = 1}^{n} \, \overline{c_{i , j}^{k}} \, x_{j}$}. Therefore
\[
\langle x_{j} \mid f_{k} \circ \sigma \circ f_{k}^{\ast} \mid x_{i} \rangle \ = \ 
\langle f_{k}^{\ast} ( x_{j} ) \mid \sigma \mid f_{k}^{\ast} ( x_{i} ) \rangle \ = \ 
\]
\[
\langle \sum_{s = 1}^{n} \, \overline{c_{j , s}^{k}} \, x_{s} \mid \sigma \mid 
\sum_{t = 1}^{n} \, \overline{c_{i , t}^{k}} \, x_{t} \rangle \ = \ 
\sum_{l = 1}^{n} \, c_{j , l}^{k} \, \lambda_{l} \, \overline{ c_{i , l}^{k} }.
\]
Similarly
\[
\langle x_{j} \mid f_{k} \circ f_{k}^{\ast} \mid x_{i} \rangle \ = \ 
\sum_{l = 1}^{n} \, c_{j , l}^{k} \, \overline{ c_{i , l}^{k} }
\]
and 
\[
\langle x_{j} \mid f_{k}^{\ast} \circ f_{k} \mid x_{i} \rangle \ = \ 
\sum_{l = 1}^{n} \, \overline{c_{l , j}^{k}} \, c_{l , i}^{k}.
\]
Therefore:
\[
\mu_{i} \ = \ \langle x_{i} \mid \tau \mid x_{i} \rangle \ = \ 
\langle x_{i} \mid \sum_{k = 1}^{m} ( f_{k} \circ \sigma 
\circ f_{k}^{\ast} ) \mid x_{i} \rangle \ = \ 
\sum_{l = 1}^{n} \lambda_{l} \, 
( \sum_{k = 1}^{m} c_{i , l}^{k} \, \overline{ c_{i , l}^{k} } ).
\]

Let the \mbox{$n \times n$} matrix $B$ be defined by:
\mbox{$b_{i , j} =$}  
\mbox{$\sum_{k = 1}^{m} c_{i , j}^{k} \, \overline{ c_{i , j}^{k} }$}.
We see that 
\begin{equation} \label{eq:vec-maj}
\mu \ = \  B \: \lambda.
\end{equation}
The summation of the elements of the $i$th row of $B$ is:
\[
\sum_{j = 1}^{n} \, b_{i , j} \ = \ 
\sum_{k = 1}^{m} \sum_{j = 1}^{n}  c_{i , j}^{k} \, \overline{ c_{i , j}^{k} } \ = \ 
\sum_{k = 1}^{m} \, \langle x_{i} \mid f_{k} \circ f_{k}^{\ast} \mid x_{i} \rangle \ = \ 1.
\]
Similarly for the summation of the elements of the $j$th column:
\[
\sum_{i = 1}^{n} \, b_{i , j} \ = \ 
\sum_{k = 1}^{m} \sum_{i = 1}^{n}  c_{i , j}^{k} \, \overline{ c_{i , j}^{k} } \ = \ 
\sum_{k = 1}^{m} \, \langle x_{j} \mid f_{k} \circ f_{k}^{\ast} \mid x_{j} \rangle \ = \ 1.
\]

We note that the matrix $B$ is bi-stochastic and conclude by Theorem 2.A.4
of~\cite{Marshall-Olkin:Inequalities} that \mbox{$\lambda \succeq \mu$}.

We have proved our claim under the assumption that $\sigma$ and $\tau$ commute.
Let us now treat the general case. There is a unitary transformation $U$ such that
\mbox{$U \circ \sigma \circ U^{\ast}$} and $\tau$ commute.
Consider the family \mbox{$g_{k} =$} \mbox{$f_{k} \circ U^{\ast}$}.
The $g_{k}$ form a bi-stochastic measurement.
We have just proven that 
\[
U \circ \sigma \circ U^{\ast} \succeq 
\sum_{k = 1}^{m} \, g_{k} \circ U \circ \sigma \circ U^{\ast} \circ g_{k}^{\ast} \ = \ 
\sum_{k = 1}^{m} \, f_{k} \circ \sigma \circ f_{k}^{\ast}.
\]
We conclude that 
\mbox{$\sigma \succeq \sum_{k = 1}^{m} \, f_{k} \circ \sigma \circ f_{k}^{\ast}$}. 

The proof of the {\em only if} direction is easier. Notations are as above.
Assume \mbox{$\sigma \succeq \tau$}.
Assume, at first, that $\sigma$ and $\tau$ commute.
By results of Hardy, Littlewood and P\'{o}lya~\cite{HLP:Inequalities}
and Birkhoff~\cite{Birkhoff:bistoch} (see Theorem 3.1.2 in~\cite{Nielsen:intro} )
there is a vector 
\mbox{$\{ p_{i} \}_{k = 1}^{m}$} of non-negative 
real numbers that sum up to $1$ and permutation matrices 
\mbox{$\{ P_{k} \}_{k = 1}^{m}$} such that \mbox{$ \mu =$}
\mbox{$\sum_{k =1}^{m} \, p_{k} \, P_{k} \lambda $}.
Considering the diagonal matrices representing $\tau$ and $\sigma$ in the basis
of their joint eigenvectors. One sees that
\mbox{$\tau =$} 
\mbox{$\sum_{k = 1}^{m} \, p_{k} \, P_{k} \circ \sigma \circ P_{k}^{\ast}$}.
The family \mbox{$f_{k} = \sqrt{p_{k}} \, P_{k}$} forms a bi-stochastic measurement
with the desired properties.
Now we want to get rid of the assumption that $\sigma$ and $\tau$ commute.
There is a unitary transformation $U$ such that $\sigma$ and 
\mbox{$\rho =$} \mbox{$U \circ \tau \circ U^{\ast}$} commute. We have
\mbox{$\sigma \succeq \rho$} and we just proved there is a bi-stochastic 
measurement \mbox{$\{ f_{k} \}_{k = 1}^{m}$} such that 
\mbox{$\rho =$} \mbox{$\sum_{k = 1}^{m} \, f_{k} \circ \sigma f_{k}^{\ast}$}.
We have \mbox{$\tau =$} 
\mbox{$\sum_{k = 1}^{m} \, U^{\ast} \circ f_{k} \circ \sigma \circ f_{k}^{\ast} \circ U$}
and the family \mbox{$\{ U^{\ast} \circ f_{k} \}$} is a suitable bi-stochastic measurement.
\end{proof}

It is known that trace preserving measurements may decrease the entropy and
therefore the bi-stochastic assumption in Theorem~\ref{the:vNresult} cannot be 
dispensed with.

\section{Entangled systems} \label{sec:entangled}
We now wish to study the effect of {\em local} measurements on local and global
states in entangled systems.
We assume each of $n$ parties, i.e., agents, 
has some piece of a quantic system. 
The pieces do not have to be similar.
Let \mbox{$\bH = \bH_{1} \otimes \bH_{2} \otimes \ldots \otimes \bH_{n}$} 
be a tensor product of $n$ finite-dimensional Hilbert spaces. 
We shall denote by \mbox{$\bG_{i}$} the tensor product of all spaces 
\mbox{$\bH_{j}$} for \mbox{$j \neq i$}.
We consider that the global system represented by $\bH$ 
is made of $n$ different parts, 
represented by $\bH_{i}$, for \mbox{$i = 1 , \ldots , n$}, 
the $i$'s part being controlled by agent $i$. 
In accordance with tradition, we assume agent $1$ is {\em Alice}.
{\em Bob} will be used as a generic name for any agent other than Alice. 
If the global state of the system is described by state 
\mbox{$\sigma : \bH \rightarrow \bH$}, the local state
of agent $i$ is described by the partial trace of $\sigma$ on \mbox{$\bG_{i}$}:
\mbox{$Tr_{\bG_{i}} ( \sigma ) :$} \mbox{$\bH_{i} \rightarrow \bH_{i}$}.
We focus here on the effect of a measurement performed by Alice on her own
state, Bob's state and the global state.

The effect of a local measurement 
\mbox{$\{ f_{k} \}_{k = 1}^{m}$}, \mbox{$f_{k} : \bH_{1} \rightarrow \bH_{1}$} 
on the global system is that of the global measurement
\mbox{$\{ g_{k} =$} \mbox{$( f_{k} \otimes {id}_{\bG_{1}} ) \}_{k = 1}^{m}$}, 
\mbox{$g_{k} : \bH \rightarrow \bH$}.
Note that, indeed, the latter is a measurement and that it is bi-stochastic iff 
the local measurement is.
Note also that the effect of the local measurement on Alice's state is as expected:
the effect of the local measurement on the local state.
\[
Tr_{\bG_{1}} ( g_{k} \circ \sigma \circ  g_{k}^{\ast} ) \ = \ 
f_{k} \circ Tr_{\bG_{1}} ( \sigma ) \circ f_{k}^{\ast}.
\]

\section{Local efficient measurement} \label{sec:localeff}
\subsection{Majorization of local states} \label{sec:majlocal}
As remarked in Section~\ref{sec:entangled}, a local measurement of Alice
acts on the global state as a global measurement would do and on
Alice's state as it would do if Alice were alone.
We conclude, by Theorem~\ref{the:expect} that the expected spectrum 
of the global state majorizes the spectrum of the initial global state and the
expected spectrum of Alice's local state majorizes the spectrum of her initial
state.
The following describes what happens to the spectrum of Bob's (or any agent different from Alice) state.
\begin{theorem} \label{the:local-other-expect}
Let $\sigma$ be a state of $\bH$ and 
\mbox{$\{ f_{k} \}_{k = 1}^{m}$}, \mbox{$f_{k} : \bH_{1} \rightarrow \bH_{1}$} 
be a local measurement of Alice. 
Let \mbox{$\{ g_{k} =$} \mbox{$ ( f_{k} \otimes {id}_{\bG_{1}} ) \}_{k = 1}^{m}$}.
For any agent \mbox{$i > 1$}, one has:
\begin{equation} \label{eq:Bobexpect}
\sum_{k = 1}^{m} \, Sp ( Tr_{\bG_{i}} ( g_{k} \circ \sigma \circ g_{k}^{\ast} ) ) \ 
\succeq \ Sp ( Tr_{\bG_{i}} ( \sigma ) ).
\end{equation}
\end{theorem}
\begin{proof}
Since \mbox{$i > 1$}, 
\mbox{$Tr_{\bG_{i}} ( g_{k} \circ \sigma \circ g_{k}^{\ast} ) =$}
\mbox{$Tr_{\bG_{i}} ( \sigma \circ g_{k}^{\ast} \circ g_{k} )$}.
By Theorem~\ref{the:sum}, 
\[
\sum_{k = 1}^{m} \, Sp ( Tr_{\bG_{i}} ( \sigma \circ g_{k}^{\ast} \circ g_{k} ) ) \ \succeq \ 
Sp ( \sum_{k = 1}^{m} \, Tr_{\bG_{i}} ( \sigma \circ g_{k}^{\ast} \circ g_{k} ) ) \ = \ 
\]
\[
Sp ( Tr_{\bG_{i}} ( \sum_{k = 1}^{m} ( \sigma \circ g_{k}^{\ast} \circ g_{k} ) ) \ = \
Sp ( Tr_{\bG_{i}} ( \sigma ) ).
\]
\end{proof}

\subsection{LOCC operations weakly increase the spectra of all local states in the majorization order}
\label{sec:increase}
\begin{theorem} \label{the:increase}
In any LOCC protocol, the spectrum of any initial local state is majorized by its expected final 
local spectrum in the majorization order.
\end{theorem}
\begin{proof}
We have shown in Section~\ref{sec:localeff} that any measurement operation
brings about, for any agent, a situation in which the expected spectrum 
majorizes the initial one. 
A local unitary operation of Alice does not change the local state of Bob 
and does not change the spectrum of her own local state. 
Classical communication does not change the global quantum state. 
We see that no step in a LOCC protocol can decrease any 
local spectrum in the majorization order.
\end{proof}

\subsection{Derivation of a generalization of 
one-half of Nielsen's characterization} \label{sec:Nielsen}
We can now derive a generalization of one half (the only if part) of Nielsen's
Theorem 1 in~\cite{Nielsen:LOCC}.

\begin{corollary} \label{the:onlyif-Nielsen}
If there is an $n$-party protocol consisting of local unitary operations, 
local generalized measurements and classical communication that, starting in a 
mixed global state $\sigma$ terminates {\em for sure}, 
i.e., with probability one, in mixed global state $\sigma'$, then, for every agent $i$,
\mbox{$ Tr_{G_{i}} ( \sigma' ) \succeq Tr_{G_{i}} ( \sigma )$}.
\end{corollary}
\begin{proof}
At each step of the protocol, we have shown that, for any agent, the initial mixed local state is 
majorized by the expected spectrum of the final mixed local state. If the final global state
is, for sure, $\sigma'$, the final mixed local states are $\rho_{i}^{\sigma'}$ and the expected
spectra are $Sp ( \rho_{i}^{\sigma'} )$. We conclude that, for every $i$, \mbox{$1 \leq i \leq n$}
one has: \mbox{$\rho_{i}^{\sigma'} \succeq \rho_{i}^{\sigma}$}, proving our claim.
\end{proof}

One may note that our results do not use Schmidt's decomposition, which is used heavily
in~\cite{Nielsen:LOCC}.

\section{Local trace preserving measurements} \label{sec:local-no}
We can now show that the result of Section~\ref{sec:vNresult} can
be extended and strengthened. 
A trace preserving bi-stochastic local measurement of Alice cannot bring about 
any additional information concerning the global state, 
Alice's own state or Bob's state. 
In fact, it leaves Bob's state unchanged.
Our claim concerning the global state and Alice's state follows directly from 
Theorem~\ref{the:vNresult} since the transformations of those states are 
bi-stochastic measurements. Let us deal with Bob's case and show that his state
is not affected by Alice's measurement.

\begin{theorem} \label{the:no-rec-Alice}
Let $\sigma$ be a state of $\bH$ and 
\mbox{$\{ f_{k} \}_{k = 1}^{m}$}, \mbox{$f_{k} : \bH_{1} \rightarrow \bH_{1}$} 
be a local measurement of Alice. 
Let \mbox{$\{ g_{k} =$} \mbox{$ ( f_{k} \otimes {id}_{\bG_{1}} ) \}_{k = 1}^{m}$}.
For any agent \mbox{$i > 1$}, one has:
\begin{equation} \label{eq:Bobno}
Tr_{\bG_{i}} ( \sum_{k = 1}^{m} \, g_{k} \circ \sigma \circ g_{k}^{\ast} ) \ = \ 
Tr_{\bG_{i}} ( \sigma ).
\end{equation}
\end{theorem}
\begin{proof}
Since $g_{k}$ is the identity on $\bH{i}$, one has:
\[
Tr_{\bG_{i}} ( g_{k} \circ \sigma \circ g_{k}^{\ast} ) \ = \ 
Tr_{\bG_{i}} ( \sigma \circ g_{k}^{\ast} \circ g_{k}^{\ast} ) 
\] 
and
\[
Tr_{\bG_{i}} ( \sum_{k = 1}^{m} \, g_{k} \circ \sigma \circ g_{k}^{\ast} ) \ = \ 
Tr_{\bG_{i}} ( \sum_{k = 1}^{m} \, \sigma \circ g_{k}^{\ast} \circ g_{k} ) \ = \ 
Tr_{\bG_{i}} ( \sigma \circ \sum_{k = 1}^{m} \, g_{k} \circ g_{k}^{\ast} ) \ = \ 
Tr_{\bG_{i}} ( \sigma ).
\]
\end{proof}

We conclude that trace preserving measurements by Alice cannot change 
the local states of any other agent. 
A bi-stochastic trace preserving measurement by Alice can only degrade 
the information contained in Alice's local state or in the global state.

\section*{Acknowledgements}
Dorit Aharonov convinced me that entanglement was fundamental and pointed
me to the right places. Discussions with her, Michael Ben-Or and Samson Abramsky
are gratefully acknowledged.

\bibliographystyle{plain}

\end{document}